\documentclass[aps,pra,superscriptaddress,footinbib,twocolumn]{revtex4-2}

\usepackage{amsmath,amsthm,amsfonts,amssymb,amscd}
\usepackage{indentfirst}
\usepackage{graphicx}
\usepackage{color}
\usepackage[american]{babel}
\usepackage{float}
\usepackage[dvipsnames]{xcolor}
\usepackage[colorlinks=true,citecolor=ForestGreen,linkcolor=magenta,urlcolor=Mulberry,hyperindex]{hyperref}
\usepackage{mathtools}
\usepackage{bm}
\usepackage{bbold}
\usepackage{units}

\newcommand{\ket}[1]{\vert#1\rangle}
\newcommand{\bra}[1]{\langle#1\vert}

\newcommand{\ketbra}[2]{\vert #1 \rangle \langle #2 \vert}

\newcommand{\id}{\mathbb{1}}
\newcommand{\set}[1]{\mathcal{#1}}

\newcommand{\Lcal}{\mathcal{L}}
\newcommand{\Hcal}{\mathcal{H}}

\newcommand{\Ecal}{\mathcal{E}}
\newcommand{\Tcal}{\mathcal{T}}
\newcommand{\Scal}{\mathcal{S}}
\newcommand{\Wcal}{\mathcal{W}}

\newcommand{\tr}{\text{\normalfont Tr}}

\newcommand{\beq}{\begin{equation}}
\newcommand{\eeq}{\end{equation}}

\newtheorem{theorem}{Theorem}
\newtheorem*{theorem*}{Theorem}
\newtheorem{definition}{Definition}

\newcommand{\iqoqi}{Institute for Quantum Optics and Quantum Information (IQOQI), Austrian Academy of Sciences, Boltzmanngasse 3, A-1090 Vienna, Austria}
\newcommand{\univie}{Vienna Center for Quantum Science and Technology (VCQ), Faculty of Physics, University of Vienna, Boltzmanngasse 5, A-1090 Vienna, Austria}
\newcommand{\todai}{Department of Physics, Graduate School of Science, The University of Tokyo, Hongo 7-3-1, Bunkyo-ku, Tokyo 113-0033, Japan}

\newcommand{\map}[1]{\widetilde{#1}} 
\renewcommand{\H}{\mathcal{H}}

\begin{document}

\title{Strict hierarchy between parallel, sequential, and indefinite-causal-order strategies for channel discrimination} 

\author{Jessica Bavaresco}
\email{jessica.bavaresco@oeaw.ac.at}
\affiliation{\iqoqi}

\author{Mio Murao}
\affiliation{\todai}
\affiliation{Trans-scale Quantum Science Institute, The University of Tokyo, Hongo 7-3-1, Bunkyo-ku, Tokyo 113-0033, Japan}

\author{Marco T\'ulio Quintino}
\affiliation{\iqoqi}
\affiliation{\univie}
\affiliation{\todai}

\date{November 18, 2021}

\begin{abstract}
We present an instance of a task of minimum-error discrimination of two qubit-qubit quantum channels for which a sequential strategy outperforms any parallel strategy. We then establish two new classes of strategies for channel discrimination that involve indefinite causal order and show that there exists a strict hierarchy among the performance of all four strategies. Our proof technique employs a general method of computer-assisted proofs. We also provide a systematic method for finding pairs of channels that showcase this phenomenon, demonstrating that the hierarchy between the strategies is not exclusive to our main example.
\end{abstract}

\maketitle

The discrimination of physical operations is a task related to the elementary ability to experimentally distinguish among different dynamics, or time evolutions, to which physical systems are subjected. From a fundamental perspective, the capacity to test and discriminate between different hypothesis lies at the core of statistical analysis and constitutes one of the pillars of the scientific method. 
From a more practical standpoint, the discrimination of physical operations comes into play in problems such as the identification of cause-effect relations~\cite{chiribella19}, computational complexity analysis of, e.g., oracle-based algorithms~\cite{deutsch92,grover98,chefles07,reitzner14}, certification of circuit elements~\cite{chiribella08-3,skotiniotis18,pereira20}, and arises naturally in tasks related to metrology~\cite{giovannetti20}. 

Within the context of quantum physics, pioneering work connecting hypothesis testing with estimation and discrimination of quantum objects dates back to Holevo and Helstrom~\cite{holevo73,helstrom69}. Most of these initial results concerned the discrimination of quantum states, while developing the concepts and methods to analyse the fundamental problem of discriminating between quantum operations---a task also referred to as \textit{quantum channel discrimination}. A plethora of interesting results on this topic has since been demonstrated~\cite{kitaev97,aharonov98,acin01,dariano01,duan07,ziman08,duan09,duan16,becker20}.  

In channel discrimination tasks that allow multiple interactions with the channels of interest, different discrimination strategies become relevant, the most common being parallel and sequential (i.e. adaptive) strategies. When considering pairs of unitary channels, optimal minimum-error discrimination has been shown to be achieved by parallel schemes~\cite{chiribella08-1}. The advantage of sequential strategies first became apparent in Ref.~\cite{harrow10}, for a task regarding two qubit-ququart entanglement-breaking channels, which are channels that cannot transmit quantum information~\cite{horodecki03}. Recent results also indicate this advantage~\cite{pirandola19,zhuang20,pereira20,rexiti20}, including its numerical observation in a discrimination task of two qubit-qubit generalized amplitude-damping channels~\cite{katariya20}.

In a related task of the discriminating of two nonsignaling bipartite channels, a more general strategy constructed from the quantum switch~\cite{chiribella13}, involving indefinite causal order, provided an advantage over causal (sequential and parallel) strategies, even allowing for perfect discrimination~\cite{chiribella12}. This phenomenon hints that indefinite causal order could be useful for the task of channel discrimination, similarly to how it has proven to be advantageous for other tasks, such as quantum computation~\cite{araujo14}, communication complexity~\cite{feix15,guerin16}, and the inversion of unknown unitary operations~\cite{quintino18}.

In this Letter, we have two main contributions to the study of channel discrimination. The first is the rigorous demonstration of an example of the advantage of sequential over parallel strategies. Our example concerns the simplest scenario of a task of channel discrimination---between a pair of qubit-qubit channels using two copies---and channels with non-zero quantum capacity---an amplitude-damping and a bit-flip channel. The second is the demonstration that strategies that involve indefinite causal order can outperform parallel and sequential strategies for the same task of channel discrimination. In order to do so, we define two new classes of
discrimination strategies that make use of indefinite causal order---which we call separable and general. Together, these results constitute a strict hierarchy between four different strategies of channel discrimination. To demonstrate our results, we develop and apply a general method of computer-assisted proofs.

The task of minimum-error channel discrimination works as follows: With probability $p_i$, Alice is given an unknown quantum channel $\widetilde{C}_i:\Lcal(\Hcal^I)\to\Lcal(\Hcal^O)$, drawn from an ensemble $\Ecal=\{p_j,\widetilde{C}_j\}_{j=1}^N$ that is known to her. Being allowed to use a finite number of copies of the channel $\widetilde{C}_i$, her task is to determine which channel she received, by performing operations on this channel and guessing the value of $i\in\{1,\ldots,N\}$. This problem is equivalent to Alice extracting the ``classical information'' $i$ which is encoded in the channel $\widetilde{C}_i$. In the simplest case of this task, when Alice is allowed to use one copy of the channel she received, the most general quantum operations that Alice could apply in her laboratory are to send part of a potentially entangled state $\rho\in\Lcal(\Hcal^I\otimes\Hcal^\text{aux})$ through the channel $\widetilde{C}_i$, and jointly measure the output with a positive operator-valued measure (POVM) $M=\{M_a\}, M_a\in\Lcal(\Hcal^O\otimes\Hcal^\text{aux})$, announcing the outcome of her measurement as her guess. Then, her probability of correctly guessing the value of $i$ is given by $p_\text{succ} \coloneqq \sum_{i=1}^N p_i \tr\left[(\widetilde{C}_i\otimes\widetilde{\id})(\rho)\,M_i\right]$, where $\widetilde{\id}$ is the identity map on $\Lcal(\Hcal^\text{aux})$. Alice can improve her chances by optimizing over the operations she applies on the unknown channel based on her knowledge of the ensemble. Her maximal probability of success is then given by $p^*_\text{succ} \coloneqq \max_{\{\rho,M\}} p_\text{succ}$.

By means of the Choi-Jamio\l{}kowski isomorphism (see footnote~\cite{fn::cj})~\cite{depillis67,jamiolkowski72,choi75} we can represent a quantum channel $\widetilde{C}$ (i.e. a completely positive, trace-preserving map) as a positive semidefinite operator $C\in\Lcal(\Hcal^I\otimes\Hcal^O)$, called its ``Choi operator'', that satisfies $C\geq0$ and $\tr_O C = \id^I$, where $\tr_X$ denotes the partial trace over $\Hcal^X$ and $\id^X$ the identity operator on $\Hcal^X$. Using Choi operators and the link product (see footnote~\cite{fn::linkprod})~\cite{chiribella09} to represent a concatenation of operators, we can rewrite the maximal probability of successful discrimination as $p^*_\text{succ}= \max_{\{\rho,M\}}\sum_{i=1}^N p_i \, C_i*\rho*M_i^T$.

In principle, Alice could apply a more general strategy by constructing the most general map that takes a quantum channel to a set of probability distributions. This map is defined by the most general set of operators $T=\{T_i\}_{i=1}^N, T_i\in\Lcal(\Hcal^I\otimes\Hcal^O)$ that respect the relation $p(i|C)=\tr(C\,T_i)$ for all Choi states of channels $C$, where $\{p(i|C)\}$ is a probability distribution. This set of operators has been characterized as a \textit{general tester}, a set $T=\{T_i\}$ that satisfies $T_i\geq0\,\forall i$ and $\sum_iT_i=\sigma\otimes\id^O$, where $\sigma\in\Lcal(\Hcal^I)$ is a quantum state~\cite{ziman08,chiribella09} (see also Appendix~\ref{app::gentesters}). Remarkably, it has been shown that every general tester has a \textit{quantum realization} in terms of states and measurements. Namely, for any strategy given by a general tester, a state and measurement that are able to implement it can always be constructed, in such a way that each tester element can be recovered as $T_i=\rho*M_i^T$.  This mathematical equivalence allows for a simpler characterization of Alice's strategies, who can now optimize over general testers $T$ to achieve a maximal probability of successful discrimination that is equivalently given by $p^*_\text{succ} = \max_{\{T\}} \sum_{i=1}^N p_i \tr\left(T_i\,C_i\right)$ (see footnote~\cite{fn::transpose}).

Now let us analyze the more interesting case in which Alice has access to two copies of the channel $C_i$. With two copies, Alice has the freedom of choosing how to concatenate these channels in order to gain more information about them.

The first and simplest option is to apply the two copies of the unknown channel in parallel, by sending a joint state $\rho\in\Lcal(\Hcal^{I_1}\otimes\Hcal^{I_2}\otimes\Hcal^\text{aux})$ through both copies of $C_i$ and then measuring the output with a POVM $M=\{M_i\}, M_i\in\Lcal(\Hcal^{O_1}\otimes\Hcal^{O_2}\otimes\Hcal^\text{aux})$, where $\Hcal^{I_1}$($\Hcal^{I_2}$) represents the input space of the first (second) copy of $C_i$, and equivalently for the output spaces. Just like in the one-copy case, this strategy can be expressed by a two-copy parallel tester, a set of operators $T^\text{PAR}=\{T_i^\text{PAR}\}$, that satisfy a number of linear constraints defined below, and that always accept a quantum realization in terms of states and measurements, according to $T^\text{PAR}_i = \rho*M_i^T$~\cite{chiribella09} (see Fig.~\ref{fig::realization}(a)). In the following we use the notation $_X A\coloneqq\tr_X A\otimes\frac{\id^X}{d_X}$ and $d_X=\text{dim}(\set{H}^X)$.

\begin{definition}[Two-copy Parallel Tester]
A parallel tester is a set of linear operators $T^\text{PAR}=\{T_i^\text{PAR}\}_{i=1}^N, T^\text{PAR}_i\in\Lcal(\Hcal^{I_1O_1I_2O_2})$ such that $T^\text{PAR}_i\geq0,\,\forall i$ and $W^\text{PAR}:=\sum_i T^\text{PAR}_i$ satisfies $\tr(W^\text{PAR})=d_{O_1}d_{O_2}$ and
\begin{equation}
    W^\text{PAR} = _{O_1O_2}W^\text{PAR}.
\end{equation}
$W^\text{PAR}$ is called a parallel process.
\end{definition}

More generally, Alice could use her two copies of $C_i$ in a sequential manner, first sending a state $\rho\in\Lcal(\Hcal^{I_1}\otimes\Hcal^{\text{aux}_1})$ through the first copy of $C_i$, next applying to the output a general channel $\widetilde{E}:\Lcal(\Hcal^{O_1}\otimes\Hcal^{\text{aux}_1})\to\Lcal(\Hcal^{I_2}\otimes\Hcal^{\text{aux}_2})$, then sending part of the output of channel $\widetilde{E}$ through the second copy of $C_i$, and finally measuring the output with a POVM $M=\{M_i\}, M_i\in\Lcal(\Hcal^{O_2}\otimes\Hcal^{\text{aux}_2})$. Analogously to the parallel case, the tester associated to this strategy---a sequential tester $T^\text{SEQ}=\{T^\text{SEQ}_i\}$ which can be expressed as $T^\text{SEQ}_i = \rho*E*M_i^T$, where $E\in\Lcal(\Hcal^{O_1}\otimes\Hcal^{\text{aux}_1}\otimes\Hcal^{I_2}\otimes\Hcal^{\text{aux}_2})$ is the Choi operator of map $\widetilde{E}$, meaning it can always be realized by quantum circuit~\cite{chiribella09} (see Fig.~\ref{fig::realization}(b))---has been characterized as

\begin{definition}[Two-copy Sequential Tester]
A sequential tester is a set of linear operators $T^\text{SEQ}=\{T_i^\text{SEQ}\}_{i=1}^N, T^\text{SEQ}_i\in\Lcal(\Hcal^{I_1O_1I_2O_2})$ such that $T^\text{SEQ}_i\geq0,\,\forall i$ and $W^\text{SEQ}:=\sum_i T^\text{SEQ}_i$ satisfies $\tr(W^\text{SEQ})=d_{O_1}d_{O_2}$ and
\begin{align}
    W^\text{SEQ} &= _{O_2}W^\text{SEQ} \\
    _{I_2O_2}W^\text{SEQ} &= _{O_1I_2O_2}W^\text{SEQ}.
\end{align}
$W^\text{SEQ}$ is called a sequential process.
\end{definition}

Parallel and sequential strategies have long been regarded as the most general strategies for channel discrimination. 
We now propose a more general strategy for channel discrimination than the sequential one, that arises from the following reasoning: In the same fashion of the definition of the general one-copy tester, we may define a \textit{general two-copy tester} as the most general set of operators $T^\text{GEN}=\{T_i^\text{GEN}\}$ that map a pair of quantum channels, represented by their Choi operators $C_A\in\Lcal(\Hcal^{I_1}\otimes\Hcal^{O_1})$ and $C_B\in\Lcal(\Hcal^{I_2}\otimes\Hcal^{O_2})$, to a valid probability distribution according to $p(i|C_A,C_B)=\tr[(C_A\otimes C_B)T_i^\text{GEN}]$. It is shown in Appendix~\ref{app::gentesters} that this definition is equivalent to

\begin{definition}[Two-copy General Tester]
A general tester is a set of linear operators $T^\text{GEN}=\{T_i^\text{GEN}\}_{i=1}^N, T^\text{GEN}_i\in\Lcal(\Hcal^{I_1O_1I_2O_2})$ such that $T^\text{GEN}_i\geq0,\,\forall i$ and $W^\text{GEN}:=\sum_i T^\text{GEN}_i$ satisfies $\tr(W^\text{GEN})=d_{O_1}d_{O_2}$ and
\begin{align}
    _{I_1O_1}W^\text{GEN} &= _{I_1O_1O_2}W^\text{GEN} \\
    _{I_2O_2}W^\text{GEN} &= _{O_1I_2O_2}W^\text{GEN} \\
    W^\text{GEN} = _{O_1}W^\text{GEN} &+ _{O_2}W^\text{GEN} - _{O_1O_2}W^\text{GEN}.
\end{align}
$W^\text{GEN}$ is called a general process.
\end{definition}
%
A general tester can be seen as the most general transformation that acts globally on a pair of independent channels, extracting probability distributions from it. In this sense, a general tester is to a pair of channels as a POVM is to a pair of states, and it can be analogously interpreted as a ``global measurement'' of a pair of channels. Differently from parallel and sequential testers, the definition of general testers does not take into account the order in which the channels may be acted upon.

Both parallel and sequential processes are particular cases of general processes (see Fig.~\ref{fig::sets}(b)). Nevertheless, the formalism of process matrices has shown that there are general processes that do not respect a definite causal order~\cite{oreshkov12,araujo15} -- which is defined as the ability of a process to be described as a parallel, sequential, or as a classical mixture of sequential processes, called ``causally separable'' process matrices, motivating the definition of our final class of testers:

\begin{definition}[Two-copy Separable Tester]
A separable tester is a set of linear operators $T^\text{SEP}=\{T_i^\text{SEP}\}_{i=1}^N$, $T^\text{SEP}_i\in\Lcal(\Hcal^{I_1O_1I_2O_2})$ such that $T^\text{SEP}_i\geq0,\,\forall i$ and $W^\text{SEP}:=\sum_i T^\text{SEP}_i$ satisfies $\tr(W^\text{SEP})=d_{O_1}d_{O_2}$ and
\begin{equation}
    W^\text{SEP} = q\,W^{1\prec2} + (1-q)W^{2\prec1},
\end{equation}
where $0\leq q\leq 1$ and $W^{1\prec2(2\prec1)}$ is a sequential process with slot $1(2)$ coming before slot $2(1)$. $W^\text{SEP}$ is called a separable process.
\end{definition}

Notice that our characterization is equivalent to imposing that $W^\text{SEQ}$, $W^\text{GEN}$, and $W^\text{SEP}$ are ordered, general, and causally separable process matrices, respectively~\cite{oreshkov12,araujo15}. In our terminology, the set of separable processes is the convex hull of the set of sequential processes whose slots follow the order $1\prec 2$ and $2\prec 1$, while parallel processes are the ones at the intersection of these two sets 
(see Fig.~\ref{fig::sets}(b)).

\begin{figure}
\begin{center}
	\includegraphics[width=\columnwidth]{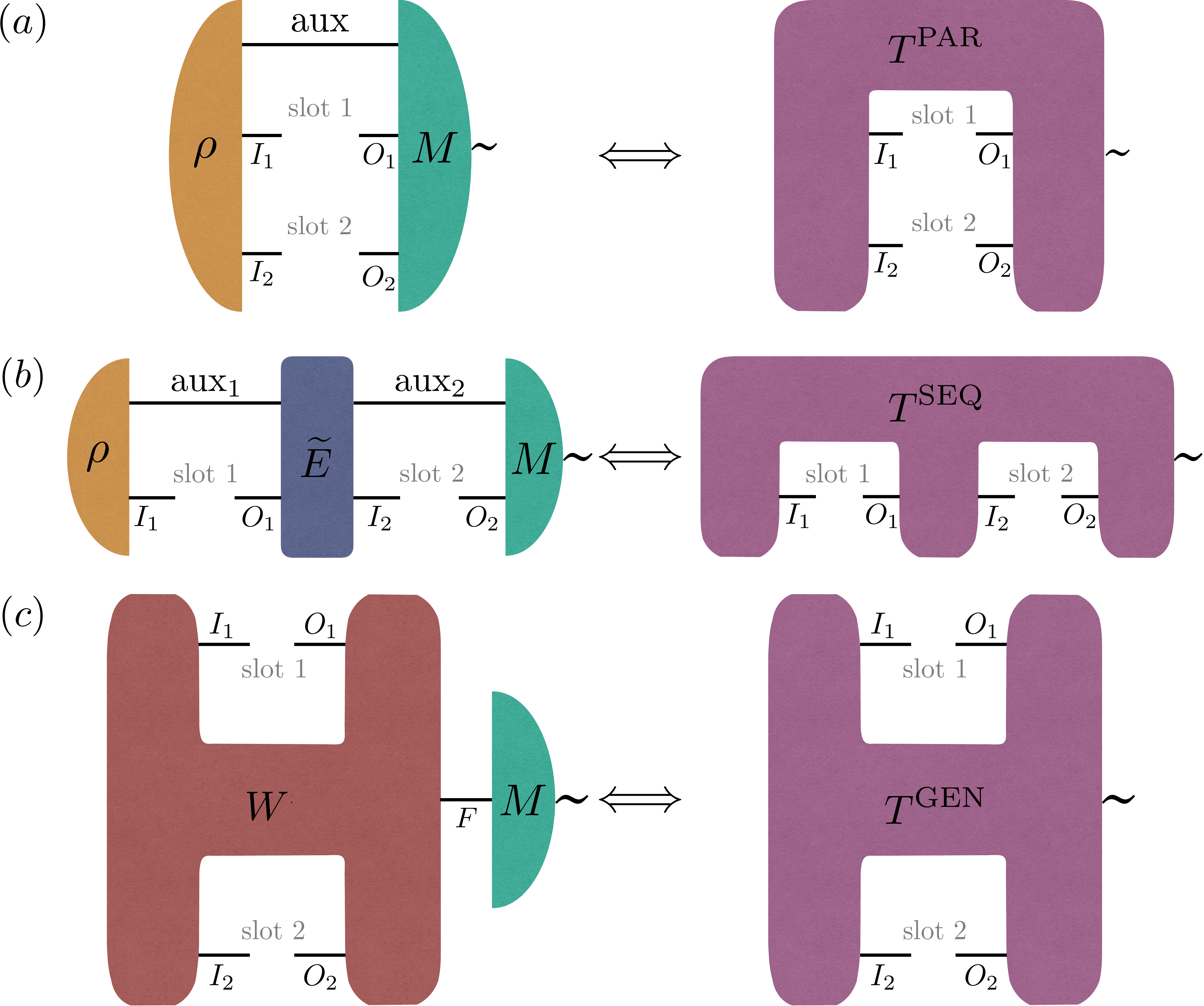}
	\caption{Schematic representation of the realization of every two-copy (a) parallel tester $T^\text{PAR}$ with a state $\rho$ and a POVM $M$, (b) sequential tester $T^\text{SEQ}$ with a state $\rho$, a channel $\widetilde{E}$, and a POVM $M$, and (c) general tester $T^\text{GEN}$ with a process matrix $W$ and a POVM $M$.}
\label{fig::realization}
\end{center}
\end{figure}

The definition of separable processes was conceived from the idea that one could plug two different channels $C_A$ and $C_B$ in the two slots of process $W^\text{SEP}$, which would then represent a mixture of a process that applies channel $C_A$ before channel $C_B$ with one that applies channel $C_B$ before $C_A$. One could then expect that this classical mixture of causal orders should not be relevant for the problem in which the two channels being plugged into the separable tester are identical: two copies of $C_i$. Nonetheless, we show that separable testers indeed provide an advantage over sequential testers, which hints at a more complicated structure of separable \textit{testers} than of separable \textit{processes} themselves. 
This advantage implies that separable testers cannot be simply realized by ordered circuits and classical randomness, and that the set of separable testers is strictly larger than the convex hull of the set of sequential testers that are ordered in different directions (see Fig.~\ref{fig::sets}(a)).

With our constructed unified framework for channel discrimination at hand, we can now define the maximal probability of successful discrimination under each of the four described strategies by allowing Alice to optimize over different classes of testers. The maximal probability of successful discrimination of a channel ensemble $\Ecal=\{p_i,C_i\}$ using two copies under strategy $\Scal\in\{\text{PAR},\text{SEQ},\text{SEP},\text{GEN}\}$ then reads
\begin{equation}
    P^\Scal\coloneqq \max_{\{T^\Scal\}} \sum_{i=1}^N p_i \tr\left(T^\Scal_i\,C_i^{\otimes2}\right).
\end{equation}
It is clear that these four strategies---parallel, sequential, separable, and general---form a hierarchy since the set of testers that they define is a superset of the previous one, in this exact order, implying the relation $P^\text{PAR}\leq P^\text{SEQ}\leq P^\text{SEP}\leq P^\text{GEN}$ for any fixed ensemble. We show that, in fact, all these three inequalities can be simultaneously strictly satisfied.

To compute the values of $P^\Scal$, we phrase the optimization problems that define it in terms of semidefinite programming (SDP). Essentially,
\begin{flalign}\label{sdp::primal}
\begin{aligned}
    \textbf{given}\ \      &\{p_i,C_i\} \\
    \textbf{maximize}\ \   &\sum_i p_i \tr\left(T^\Scal_i\,C_i^{\otimes 2}\right) \\ 
    \textbf{subject to}\ \ &\{T^\Scal_i\}\ \text{is a tester with strategy}\ \Scal.
\end{aligned}&&
\end{flalign}

This problem can be equivalently solved by its dual problem:
\begin{flalign}\label{sdp::dual}
\begin{aligned}
    \textbf{given}\ \ &\{p_i,C_i\} \\
    \textbf{minimize}\ \   &\lambda \\ 
    \textbf{subject to}\ \  &p_i\,C_i^{\otimes2}\leq \lambda\,\overline{W}^\Scal \ \ \forall\,i,
\end{aligned}&&
\end{flalign}
where $\overline{W}^\Scal$ lies in the dual affine of the set of processes $\mathcal{W}^\Scal$ (see footnote~\cite{fn::affinesep}), as demonstrated in the Appendix~\ref{app::sdp}. The dual problem can also be straightforwardly phrased as an SDP by absorbing the coefficient $\lambda$, as explained in Appendix~\ref{app::sdp}.

SDPs can be solved by efficient numerical packages which, despite being in practice accurate, suffer from imprecision that arise from the use of floating-point variables~\cite{floating_guide,floating_wikipedia}. In order to overcome this issue, we provide in the Appendix~\ref{app::compassis} an algorithm for computer-assisted proofs (see~\cite{peyrl08,rump10} for other examples). Using our method, 
we obtain rigorous upper and lower bounds for $P^\Scal$, arriving at a result that has the same mathematical rigor as an analytical proof.

\begin{theorem}
In the simplest instance of a channel discrimination task using $k=2$ copies, i.e., discrimination between $N=2$ qubit-qubit channels, there exist ensembles for which the maximal probability of successful discrimination of parallel, sequential, separable, and general strategies obey the strict hierarchy
\begin{equation}
    P^\text{PAR}<P^\text{SEQ}<P^\text{SEP}<P^\text{GEN}.
\end{equation}
\end{theorem}

\noindent\textit{Sketch of the proof.} 
The proof is constructive and considers the channel ensemble composed by $p_1=p_2=\nicefrac{1}{2}$, an amplitude-damping channel $\widetilde{C}_\text{AD}$ (see footnote~\cite{fn::ampdamp}) with damping parameter $\gamma=\nicefrac{67}{100}$, and a bit-flip channel $\widetilde{C}_\text{BF}$ (see footnote~\cite{fn::bitflip}) with flipping parameter $\eta=\nicefrac{87}{100}$. We start by applying standard numerical packages to solve the primal SDP~\eqref{sdp::primal} and obtain an ansatz for the optimal tester of each discrimination strategy. From the numerically imperfect ansatz, we construct a valid tester, following the steps of Algorithm 2 in the Appendix~\ref{app::compassis}. We then compute the probability of successful discrimination with this valid tester, which provides a rigorous lower bound for the maximal probability of success. To calculate a rigorous upper bound, we repeat this procedure, now taking as ansatz the numerical solution of the dual problem~\eqref{sdp::dual} for a dual affine process, and following the steps of Algorithm 1 in the Appendix~\ref{app::compassis}. Applying this method, we computed the following bounds:
$
{\frac{8346}{10000}< P^\text{PAR}<\frac{8347}{10000}}$,
$
{\frac{8446}{10000}< P^\text{SEQ}<\frac{8447}{10000}}$,
$
{\frac{8486}{10000}< P^\text{SEP}<\frac{8487}{10000}}$, and
$
{\frac{8514}{10000}< P^\text{GEN}<\frac{8515}{10000}}$.
The clear gap between the upper bound of one strategy and the lower bound of the next concludes the proof. \hfill$\square$

\begin{figure}
\begin{center}
	\includegraphics[width=\columnwidth]{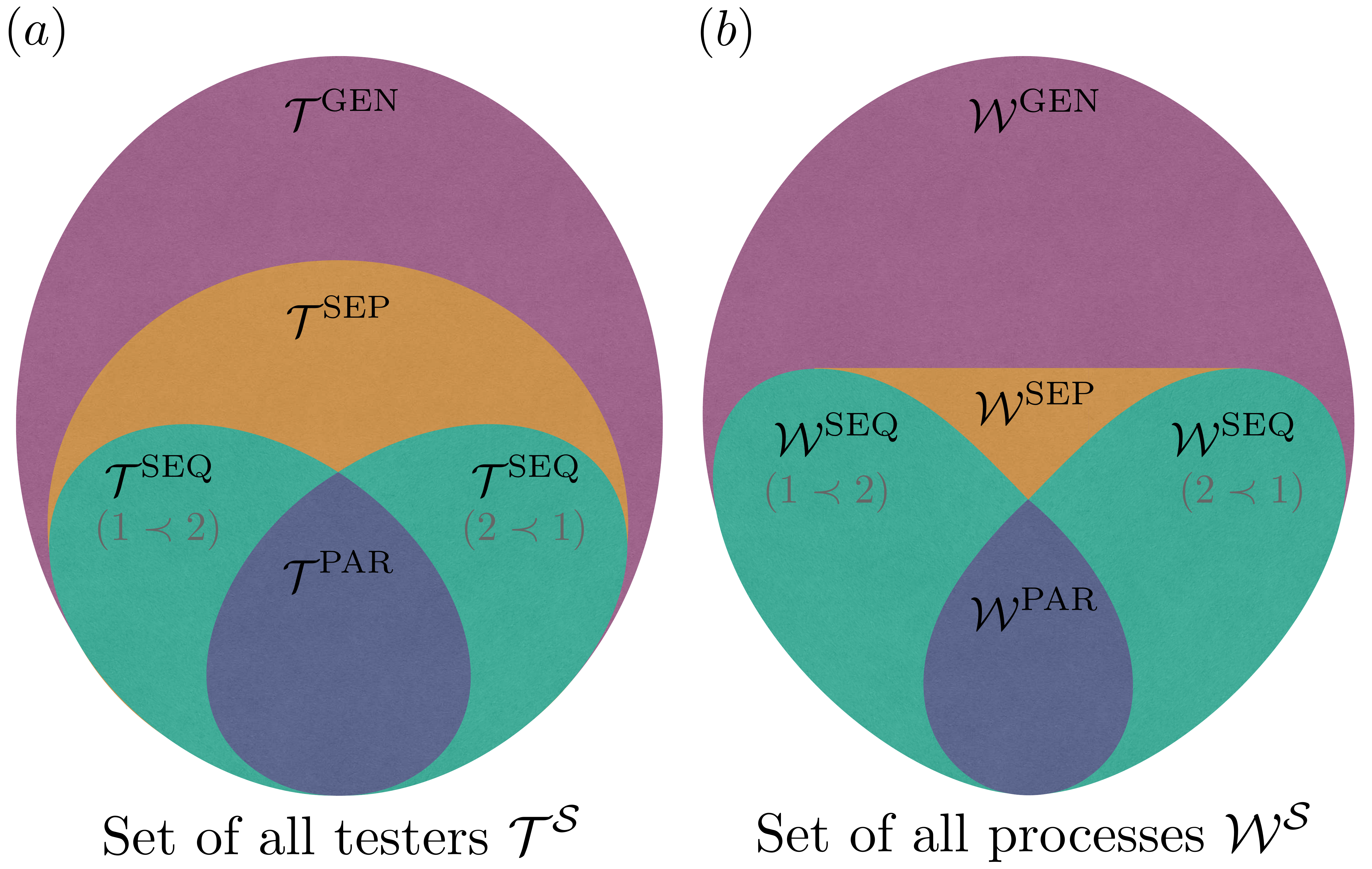}
	\caption{Graphical representation of the nesting relations between (a) the sets of all testers $\Tcal^\Scal\coloneqq\{T^\Scal;T^\Scal=\{T^\Scal_i\}\}$ and (b) the sets of all processes $\Wcal^\Scal\coloneqq\{W^\Scal;W^\Scal=\sum_iT^\Scal_i\}$, where $\Scal\in\{\text{PAR},\text{SEQ},\text{SEP},\text{GEN}\}$ represents parallel, sequential, separable, or general strategies.}
\label{fig::sets}
\end{center}
\end{figure}

Similar gaps can also be found for different ensembles of amplitude-damping and bit-flip channels, and also for ensembles of two amplitude-damping channels, a problem which has been previously studied~\cite{pirandola19,zhuang20,pereira20,rexiti20,katariya20}. Moreover, this phenomenon is not particular to these channels. We have constructed a simple method of sampling pairs of quantum channels that present a gap between all four strategies, for the case of qubit-qubit channels, in approximately $94\%$ of the rounds. See the Appendix~\ref{app::randchannels} for more details.

Having demonstrated the theoretical advantage of these strategies, we would now like to discuss their potential implementation. As already mentioned, for the case of parallel and sequential strategies, it is known that, from every tester, one can construct in an algorithmic manner a state, a channel, and a measurement that constitute a quantum realization for each tester element~\cite{chiribella09}. Therefore, these testers can be physically implemented with quantum circuits, as depicted on Fig.~\ref{fig::realization}(a) and (b).  

For the case of general testers, however, given a tester $T^\text{GEN}=\{T_i^\text{GEN}\}$, we can claim that it can be realized by a process $W\coloneqq \sum_{i=1}^N T_i^\text{GEN}\otimes\ketbra{i}{i}^F\in\Lcal(\Hcal^{I_1O_1I_2O_2}\otimes\Hcal^{F})$, where $\Hcal^{F}$ represents the Hilbert space of a system in the common future of the slots 1 and 2 of $T^\text{GEN}$, and a POVM $M=\{M_i\}_i,\,M_i=\ketbra{i}{i}\in\Hcal^{F}$. Each general tester element is recovered by $T_i^\text{GEN}=W*M_i^T$. A \textit{quantum} realization of $T^\text{GEN}$ would then depend on the ability to physically implement any process matrix $W$, as depicted in Fig.~\ref{fig::realization}(c). Unfortunately, at this point, the physical implementation of general process matrices remains an open question. 

For the case of separable testers, however, a physical implementation is known. Similar to the general case, every separable tester can be constructed from a process $W \coloneqq \sum_{i=1}^N T_i^\text{SEP}\otimes\ketbra{i}{i}^F$ and a measurement given by $M_i\coloneqq\ketbra{i}{i}^{F}$. However, when constructed from separable testers, the process $W$ always satisfies the condition that $\tr_F W = \sum_{i=1}^N T_i^\text{SEP} = W^\text{SEP}$, that is, they are (potentially nonseparable) processes that become separable when the future space is traced out. Such processes always lead to separable strategies, and can be used to realize every separable tester. Remarkably, these processes have recently been shown by Ref.~\cite{wechs21} to be realized by circuits that employ a coherent quantum control of causal orders, implying that all separable strategies, including the ones that we have shown to be advantageous over sequential strategies, can be physically implemented. One example of such a process that only leads to separable testers is the well-studied quantum switch~\cite{chiribella13}. Notice that the testers that can be generated from the quantum switch are an instance of separable testers that are not in the convex hull of sequential testers that are ordered in different directions, meaning that they could potentially be advantageous when compared to sequential strategies. 
Nevertheless, we have not been able to construct an example of a discrimination task for which testers generated by the quantum switch are advantageous.

\textit{Conclusions.} We have demonstrated a new example of the advantage of sequential over parallel strategies for a task of minimum-error discrimination between two qubit-qubit channels. We also established two new classes of strategies that involve indefinite causal order and showed that they can outperform causal ones. Moreover, we proved a strict hierarchy between these four classes of discrimination strategies. Our main example concerns the discrimination of an amplitude-damping and a bit-flip channel; however, we showed that this phenomenon is not unique, by presenting a simple method of constructing pairs of channels that, with very high probability, respect this strict hierarchy. The main technique developed in this paper was a method of computer-assisted proofs, that finds immediate application in a plethora of physics problems that currently rely on numerical optimization. We hope that this method can contribute to paving the way to more rigorous numerical proofs in quantum information science. It is furthermore our hope that our demonstration of the theoretical advantage of indefinite causal order for channel discrimination will further motivate the investigation of the potential implementation of general processes.

\textit{Acknowledgments.} We are thankful to Alastair Abbott and Simon Milz for interesting discussions and to Mateus Ara\'ujo, Nicolai Friis and Cyril Branciard for comments on the manuscript. 
J.B. would like to thank the hospitality of the Murao Group and of The University of Tokyo. J.B. acknowledges the Austrian Science Fund (FWF) through the START project Y879-N27 and the Zukunftskolleg project ZK03. 
M.M. is supported by the MEXT Quantum Leap Flagship Program (MEXT Q-LEAP) Grant Numbers JPMXS0118069605 and JPMXS0120351339 and by the Japan Society for the Promotion of Science (JSPS) through the KAKENHI grants 17H01694, 18H04286, and 21H03394. 
M.T.Q. acknowledges the Austrian Science Fund (FWF) through the SFB project BeyondC (subproject No. F7103), a grant from the Foundational Questions Institute (FQXi) as part of the Quantum Information Structure of Spacetime (QISS) Project (qiss.fr). The opinions expressed in this publication are those of the authors and do not necessarily reflect the views of the John Templeton Foundation. This project has received funding from the European Unions Horizon 2020 research and innovation program under the Marie Skłodowska-Curie Grant Agreement No. 801110. It reflects only the authors’ view; the EU Agency is not responsible for any use that may be made of the information it contains. The Erwin Schrödinger Center for Quantum Science \& Technology (ESQ) has received funding from the Austrian Federal Ministry of Education, Science and Research (BMBWF).

All our code is available in an online repository~\cite{githubMTQ} and can be freely used, edited, and distributed.



%


\onecolumngrid 
\appendix

\setcounter{equation}{0}
\setcounter{theorem}{0}
\setcounter{definition}{0}
\setcounter{figure}{0}
\setcounter{table}{0}

\section*{APPENDIX}

\setcounter{figure}{2}

Here we present support material that complements the main text. It is structured as follows: Appendix~\ref{app::gentesters}. Characterization theorem for general testers, Appendix~\ref{app::sdp}. Semidefinite programming and dual affine spaces, Appendix~\ref{app::compassis}. Computer-assisted proofs, and Appendix~\ref{app::randchannels}. Sampling general channels and the typicality of the hierarchy between discrimination strategies.

\section{Characterization theorem for general testers}\label{app::gentesters}
\setcounter{theorem}{1}

We starting by demonstrating, for sake of completeness, the characterization of general one-copy testers that was presented in the main text, in the language of our paper. This result is already known and follows from Ref.~\cite{chiribella09}.

\begin{theorem}\label{thm::onecopy_charac}
    Let $T=\{T_i\}_{i=1}^N$, $T_i\in\Lcal(\Hcal^I\otimes\Hcal^O)$, called a general one-copy tester, be the most general set of operators that satisfy the relation
    \begin{equation}
        p(i|C) = \tr\left(T_i\,C\right), 
    \end{equation}
    for all Choi operators of quantum channels $C\in\Lcal(\Hcal^I\otimes\Hcal^O)$, where $\{p(i|C)\}$ is a set of probability distributions. Let $W\coloneqq\sum_iT_i$. Then, $T=\{T_i\}$ is a set of operators that satisfy
    \begin{align}
        T_i &\geq 0 \ \ \forall\,i \\
        \tr(W) &= d_O \label{eq::1} \\
        W &=  {_O} W. \label{eq::2}
    \end{align}
\end{theorem}

\begin{proof}
    In order to guarantee that $\{p(i|C)\}$ is a valid probability distribution, two conditions must be imposed: positivity and normalization.
    
    Positivity:
    \begin{equation}
        p(i|C) = \tr\left(T_i\,C\right) \geq 0 \ \ \forall\,i,C\geq0 \ \ \iff \ \ T_i\geq 0 \ \ \forall\,i.
    \end{equation}
    
    Normalization:
    \begin{equation}
        \sum_i p(i|C) = \tr(\sum_i T_i\,C) = \tr\left(W\,C\right) = 1 \ \ \forall\,\text{channels } C,
    \end{equation}
    where $C$ is the Choi operator of a quantum channel, and therefore of a trace-preserving map, which can be parametrized as $C = X - _{O}X + \frac{\id}{d_O}$, where $X$ is a self-adjoint operator, using the same technique as in Appendix B of Ref.~\cite{araujo15}. Then,
    \begin{equation}
        \tr[W\,(X- _{O}X+\frac{\id}{d_O})] = 1 \ \ \forall\,\text{self-adjoint } X.
    \end{equation}
    We can split this in two cases: $X=0$ and $X\not=0$.
    
    For $X=0$:
    \begin{equation}
        \tr[W\,(X- _{O}X+\frac{\id}{d_O})] = \frac{\tr(W)}{d_O} = 1 \ \ \iff \ \ \tr(W) = d_O.
    \end{equation}
    
    For $X\not=0$:
    \begin{align}
        \tr[W\,(X- _{O}X+\frac{\id}{d_O})] = \tr\left[W\,(X- _{O}X)\right] + 1 &= 1  \ \ \forall\, X\not=0 \\
       \iff \tr\left[W\,(X- _{O}X)\right] &= 0 \ \ \forall\, X\not=0  \label{eq::selfdual1} \\
       \iff \tr\left[(W- _{O}W)\,X\right] &= 0 \ \ \forall\, X\not=0 \ \ \iff W- _{O}W = 0. \label{eq::selfdual2}
    \end{align}
The equivalence between Eqs.~\eqref{eq::selfdual1} and~\eqref{eq::selfdual2} is given by the self-duality of the `trace-and-replace' map, namely $\tr[W\,_{O}X]=\tr[_{O}W\,X]$.
    
    Together, conditions $\tr(W) = d_O$ (Eq.~\eqref{eq::1}) and $W= _{O}W$ (Eq.~\eqref{eq::2}) imply that $W$ can be written as $W=\sigma\otimes\id^O$, where $\sigma\in\Lcal(\Hcal^I)$ is a normalized quantum state.
\end{proof}

Now we prove a new characterization theorem, the one of general two-copy testers. In this case, we will need additional hypotheses. One is the hypothesis that a tester may not only be able to act on two copies of the same channel but also be able to act on two different, independent channels. This hypothesis is physically motivated in the sense that, if a general tester is a device in a quantum lab that can act on two copies of the same channel, then one should also be able to plug in two different channels and have it perform a meaningful physical operation. The second is that these channels should be allowed to also act on auxiliary, potentially entangled, systems, and when a general tester acts upon part of these channels, the operation it performs should still result in a valid probability distribution. This last hypothesis is automatically satisfied in the one-copy case. 

Formally, we have:
    
\begin{theorem}
    Let $T^\text{GEN}=\{T^\text{GEN}_i\}_{i=1}^N$, $T^\text{GEN}_i\in\Lcal(\Hcal^{I_1}\otimes\Hcal^{O_1}\otimes\Hcal^{I_2}\otimes\Hcal^{O_2})$, called a general two-copy tester, be the most general set of operators that satisfy the relation
    \begin{equation}
        p(i|C_A,C_B,\rho_{AB}) = \tr\left[(T^\text{GEN}_i\otimes\rho_{AB})(C_A\otimes C_B)\right], 
    \end{equation}
    for all Choi operators of quantum channels $C_A\in\Lcal(\Hcal^{I_1}\otimes\Hcal^{O_1}\otimes\Hcal^{\text{aux}_1})$ and $C_B\in\Lcal(\Hcal^{I_2}\otimes\Hcal^{O_2}\otimes\Hcal^{\text{aux}_2})$, and for all quantum states $\rho_{AB}\in\Lcal(\Hcal^{\text{aux}_1}\otimes\Hcal^{\text{aux}_2})$, where $\{p(i|C)\}$ is a set of probability distributions. Let $W^\text{GEN}\coloneqq\sum_i T^\text{GEN}_i$. Then, $T^\text{GEN}=\{T^\text{GEN}_i\}$ is a set of operators that satisfy
    \begin{align}
    T^\text{GEN}_i&\geq 0 \ \ \forall\,i \label{eq::0} \\
    \tr(W^\text{GEN})&=d_{O_1}d_{O_2} \label{eq::a}\\
    _{I_1O_1}W^\text{GEN} &= _{I_1O_1O_2}W^\text{GEN} \label{eq::b}\\
    _{I_2O_2}W^\text{GEN} &= _{O_1I_2O_2}W^\text{GEN} \label{eq::c}\\
    W^\text{GEN} = _{O_1}W^\text{GEN} &+ _{O_2}W^\text{GEN} - _{O_1O_2}W^\text{GEN}. \label{eq::d}
    \end{align}
\end{theorem}

\begin{proof}
    Again, in order to guarantee that $\{p(i|C_A,C_B)\}$ is a valid probability distribution, the conditions of positivity and normalization must be imposed.
    
    Positivity:
    \begin{equation}
        p(i|C_A,C_B,\rho_{AB}) = \tr\left[(T^\text{GEN}_i\otimes\rho_{AB})(C_A\otimes C_B)\right] \geq 0 \ \ \forall\,i,C_A\geq 0,C_B\geq 0,\rho_{AB}\geq0 \ \ \iff \ \ T^\text{GEN}_i\geq 0 \ \ \forall\,i.
    \end{equation}
    
    Normalization:
    \begin{align}
    \begin{split}\label{eq::normW}
        \sum_i p(i|C_A,C_B,\rho_{AB}) = \tr[(\sum_i T^\text{GEN}_i\otimes\rho_{AB})(C_A\otimes C_B)] = \tr\left[(W^\text{GEN}\otimes\rho_{AB})(C_A\otimes C_B)\right] = 1& \\
        \forall\,\text{channels } C_A,C_B \text{ and states } \rho_{AB}.&
    \end{split}
    \end{align}
    Notice that condition Eq.~\eqref{eq::normW} is exactly the normalization condition that, in Appendix B of Ref.~\cite{araujo15}, defines $W^\text{GEN}$ as a bipartite process matrix. Hence, it immediately follows from the proof contained therein that $W^\text{GEN}$ must respect Eqs.~\eqref{eq::a}-\eqref{eq::d}.
\end{proof}

Intuitively, Eq.~\eqref{eq::a} can be understood as the constraint that guarantees the non-negativity of the elements of the probability distributions, while Eqs.~\eqref{eq::a}-\eqref{eq::d} guarantee the normalization of the probability distributions. Equations~\eqref{eq::b} and~\eqref{eq::c} guarantee a local ordering of the inputs and outputs within each slot. Physically, these equations can be understood as the constraints that forbid local time loops. The last constraint, in Eq.~\eqref{eq::d}, can be physically understood as the constraint that forbids global time loops for occurring, which would allow one slot to exploit the channels that connect it to the second slot to feed information to its own past. A more in-depth discussion of the physical consequences of these constraint is provided in Ref.~\cite{oreshkov12}. 

\section{Semidefinite programming formulation and dual affine spaces }\label{app::sdp}
\setcounter{definition}{4}

In this section we present a method to obtain a dual problem formulation for a class of convex optimization problems which covers the SDP presented in our main text. This method employs ideas and techniques first presented in Ref.~\cite{chiribella16}.

A subset of linear operators $\mathcal{W}\subseteq \mathcal{L}(\H)$  is said to be affine if for every set of real numbers $\{w_i\}_i$ respecting $\sum_i w_i=1$, and for every subset $\{W_i\}_i\subseteq\mathcal{W}$ we have that $\left(\sum_i w_i W_i\right) \in \mathcal{W}$. 

\begin{definition}[Dual affine space~\cite{chiribella16}.]
	Let $\mathcal{W}\subseteq \mathcal{L}(\H)$ be a set of linear operators. The dual affine space $\overline{\mathcal{W}}$ of $\mathcal{W}$ is defined via
    \begin{equation}
	    \overline{W}\in \overline{\mathcal{W}} \; \text{ when } \; \tr(\overline{W}\, W)=1, \; \forall \, W \in \mathcal{W} .
    \end{equation}
\end{definition}

\begin{figure}
\begin{center}
	\includegraphics[width=\columnwidth]{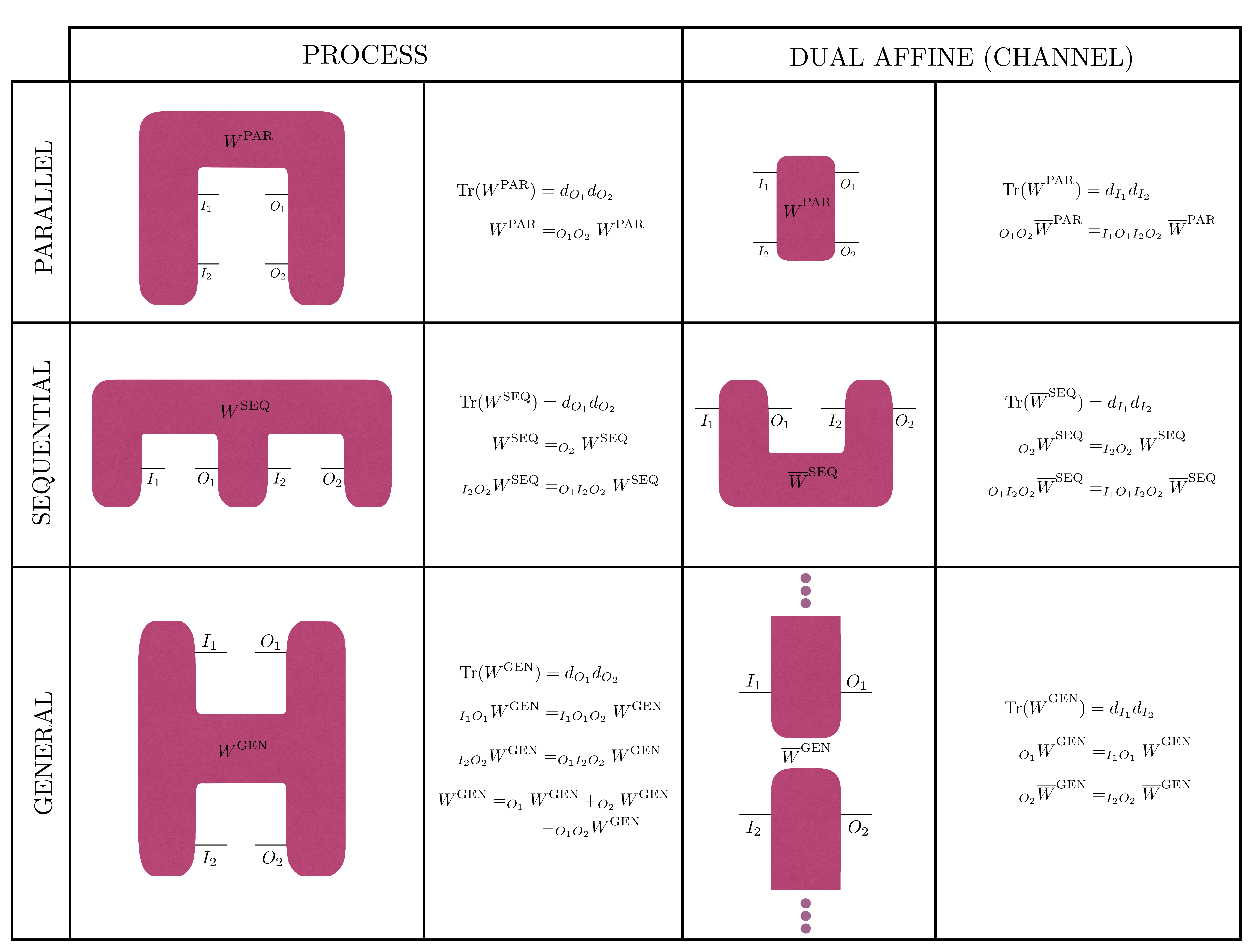}
	\caption{Normalization constraints for parallel, sequential, and general two-slot processes, and for their dual affine spaces, which correspond to the normalization constraints for bipartite channels, bipartite channels with memory, and bipartite no-signalling channel respectively. Note that the dual affine space of a set $\mathcal{W}^\Scal$ may be intuitively visualized as the largest set of `quantum objects' $\overline{\mathcal{W}}^\Scal$ such that `connecting' objects from $\mathcal{W}^\Scal$ to objects from $\overline{\mathcal{W}}^\Scal$ always lead to the scalar number $1$.}
\label{fig::affineduals}
\end{center}
\end{figure}

If $\mathcal{W}\subseteq\mathcal{L}(\H)$ is the set of all quantum states, i.e., positive semidefinite operators $W\in\mathcal{L}(\H)$ such that $\tr(W)=1$, the only operator $\overline{W}$ such that $\tr(\overline{W} \, W)=1, \; \forall W\in\mathcal{W}$ is the identity operator. Hence, the dual affine space of set of quantum states has a single element which is the identity operator $\id$ and corresponds to the normalisation constraint for quantum measurements.

If $\mathcal{W}^\text{PAR}\subseteq\mathcal{L}(\H^I\otimes \H^O)$, where $\H^I = \bigotimes_{i=1}^k I_i$ and $\H^O = \bigotimes_{i=1}^k O_i$ stands for the set of all parallel processes, i.e., positive semidefinite operators that can be written as $W^\text{PAR}=\sigma^I \otimes \id^O$, with $\tr(\sigma)=1$, one can check that its dual affine space is given by a set of linear operators $\overline{W}^\text{PAR}$ respecting $\tr_O \overline{W}^\text{PAR}=\id^I$, which is the set of quantum channels without the positivity condition.

If $\mathcal{W}^\text{SEQ}\subseteq\mathcal{L}(\H^I\otimes \H^O)$ stands for the set of all sequential processes, 
Ref.~\cite{chiribella16} shows that its dual affine space  $\overline{\mathcal{W}}^\text{SEQ}$ is given by the set of Choi operators of  $k$-partite channels with memory%
\footnote{Note that a $k$-partite channel with memory is formally equivalent to a quantum comb with $k-1$ slots \cite{chiribella09}.} \cite{kretschmann05} %
without the positivity constraint.
In particular, for the two-slot case, an operator  $\overline{W}^\text{SEQ}\in\mathcal{L}(\H^{I_1}\otimes\H^{O_1}\otimes\H^{I_2}\otimes\H^{O_2})$ belongs to the dual affine space of the sequential processes if and only if $\overline{W}^\text{SEQ}$ respects
\begin{align}
    _{O_2}\overline{W}^\text{SEQ}&=_{I_2O_2}\overline{W}^\text{SEQ} \\
    _{O_1I_2O_2}\overline{W}^\text{SEQ}&=_{I_1O_1I_2O_2}\overline{W}^\text{SEQ} \\
    \tr(\overline{W}^\text{SEQ})&=d_{I_1}d_{I_2}.
\end{align}

If $\mathcal{W}^\text{GEN}\subseteq\mathcal{L}(\H^I\otimes \H^O)$ stands for the set of all general processes, 
Ref.~\cite{chiribella16} shows that its dual affine space  $\overline{\mathcal{W}}^\text{GEN}$ is given by the set of Choi operators of  $k$-partite no-signalling channels \cite{beckman01,eggeling02} without the positivity constraint.
In particular, for the two-slot case, an operator $\overline{W}^\text{GEN}\in\mathcal{L}(\H^{I_1}\otimes\H^{O_1}\otimes\H^{I_2}\otimes\H^{O_2})$ belongs to the dual affine space of the general processes if and only if $\overline{W}^\text{GEN}$ respects
\begin{align}
    _{O_2}\overline{W}^\text{GEN}&=_{I_2O_2}\overline{W}^\text{GEN} \\
    _{O_1}\overline{W}^\text{GEN}&=_{I_1O_1}\overline{W}^\text{GEN} \\
    \tr(\overline{W}^\text{GEN})&=d_{I_1}d_{I_2}.
\end{align} 

We have summarized the normalization constraints of parallel, sequential, and general processes and their respective dual affine spaces in Fig.~\ref{fig::affineduals}.

We now describe a method for obtaining the dual formulation of the SDPs presented in this paper based on the concept of dual affine spaces. In the main text we have defined the primal optimization problem as

\begin{flalign}
\begin{aligned}
    \textbf{given}\ \      &\{p_i,C_i\} \\
    \textbf{maximize}\ \  
    &\sum_i p_i \tr \left(T_i^\Scal \, C_i^{\otimes 2} \right) \\ 
    \textbf{subject to} \ \ &\{T_i^\Scal\}\ \in \Tcal^\Scal, 
\end{aligned}&&
\end{flalign}

where $\Tcal^\Scal$ is the set of all testers with strategy $\Scal$. This problem can also be written as

\begin{flalign} \label{SDP:primal2}
\begin{aligned}
    \textbf{given}\ \      &\{p_i,C_i\} \\
    \textbf{max}\ \   &\sum_i p_i \tr\left(T_i^\Scal\,C_i^{\otimes 2}\right) \\ 
    \textbf{s.t.}\ \ &T_i^\Scal\geq0 \\
    &\sum_i T_i^\Scal \in \mathcal{W}^\Scal
\end{aligned}&&
\end{flalign}
where $\mathcal{W}^\Scal$ is set of all processes with strategy $\Scal$.

We start this section by considering the above optimization problem for the case where the set $\mathcal{W}^\Scal$ is affine, which is the case for parallel, sequential, and general processes. For these strategies, we do not need to restrict ourselves to the case of $k=2$ copies of the input channel $C_i$ but the method applies for any $k\in\mathbb{N}$. Note that the normalization constraints of separable processes do not form an affine set, for which reason the case of separable testers will be tackled later. We also point that the for $k>2$, the definition of $k$-slots separable processes have several nuances and there is still no consensus on a single definition \cite{wechs19}.

For finite dimensions, if $\mathcal{W}$ is an affine set we have  that $\overline{\overline{\mathcal{W}}}=\mathcal{W}$, i.e., the dual affine space of the dual affine space of $\mathcal{W}$ is simply $\mathcal{W}$. Hence, for cases where $\mathcal{W}$ is affine, the primal SDP presented in Eq.~\eqref{SDP:primal2} can be written as:

\begin{flalign}
\begin{aligned}
    \textbf{given}\ \      &\{p_i,C_i\} \\
    \textbf{max}\ \   &\sum_i p_i \tr\left(T_i^\Scal\,C_i^{\otimes k}\right) \\ 
    \textbf{s.t.}\ \ & T_i^\Scal\geq0 \\
    &W^\Scal:=\sum_i T_i^\Scal \\
    & \tr(W^\Scal \overline{W}^\Scal)=1,\ \  \forall \, \overline{W}^\Scal\in   \overline{\mathcal{W}}^\Scal,
\end{aligned}&&
\end{flalign}
a formulation which has infinitely many constraints 
$\left[\tr(W^\Scal \; \overline{W}^\Scal)=1, \forall\, \overline{W}^\Scal\in \overline{\mathcal{W}}^\Scal \right]$. These infinitely many constraints can be made finite by writing 
$\left[\tr(W^\Scal \; \overline{W}_j^\Scal)=1, \forall\, j\right]$ where $\{\overline{W}_j^\Scal\}_j$ is an affine basis for $\overline{\mathcal{W}}^\Scal$, i.e., every $\overline{W}^\Scal\in\overline{\mathcal{W}}^\Scal$ can be written as $\overline{W}^\Scal=\sum_j w_j \overline{W}_j^\Scal$ for a set of coefficients $\{w_j\}_j$ respecting $\sum_j w_j=1$. The Lagrangian of the maximization problem can then be written as

\begin{align}
	L &= \sum_{i} p_i \tr\left(C_i^{\otimes k} \; T_i^\Scal\right) + \sum_i \tr\left(T_i^\Scal \Gamma_i\right) + \sum_{j} \left[ 1-\tr\left(\sum_i T_i^\Scal \overline{W}_j^\Scal\right)\right] \lambda_j.
\end{align}

Hence, if $\Gamma_i\geq0$ and $\{T_i^\Scal\}_i$ is a tester, $L\geq \sum_{i} p_i \tr(C_i^{\otimes k} \; T_i^\Scal)$. By re-arranging terms, the Lagrangian can be written as
\begin{align}
L 
&= \sum_i \tr\left[T^\Scal_i\left(p_i C_i^{\otimes k} + \Gamma_i - \sum_j \overline{W}_j^\Scal\lambda_j \right) \right]  + \sum_j\lambda_j.
\end{align}
We then arrive at the dual problem by taking the supremum of the Lagrangian over the primal variables $\{T^\Scal_i\}_i$. Finally, the solution of the dual problem will be given by the minimization over the dual variables $\{\Gamma_i\}_i$ and $\{\lambda_i\}_i$ under the constraint that $\Gamma_i\geq 0, \forall\,i$. The dual problem can be written as

\begin{flalign}
\begin{aligned}
    \textbf{given}\ \      &\{p_i,C_i\} \\
    \textbf{minimize}\ \   &\sum_j \lambda_j \\ 
    \textbf{s.t.}\ \ & \Gamma_i\geq0 \ \ \forall\, i \\
    & p_i C_i^{\otimes k} + \Gamma_i + \sum_j \lambda_j \overline{W}_j^\Scal=0, \ \ \forall\, i
\end{aligned}&&
\end{flalign}

Removing the dummy variables $\{\Gamma_{i}\}$, we obtain

\begin{flalign}	
\begin{aligned}
    \textbf{given}\ \      &\{p_i,C_i\} \\
    \textbf{min}\ \   &\sum_j \lambda_j \\ 
    \textbf{s.t.}\ \ 
    & p_i C_i^{\otimes k} \leq \sum_j \lambda_j \overline{W_j}^\Scal \ \  \forall\, i .
\end{aligned}&&
\end{flalign}

The requirement of having an affine basis $\{\overline{W}_j^\Scal\}_j$ can be dropped by defining $\lambda:=\sum_j \lambda_j$ and $\overline{W}^\Scal:=\sum_j \frac{\lambda_j  \overline{W_j}^\Scal}{\lambda}$ and noting that, by construction, for any choice of coefficient $\lambda_j$, $\overline{W}^\Scal$ is an affine combination of the affine basis elements $\{\overline{W}_j^\Scal\}_i$, hence $\overline{W}^\Scal$ necessarily belongs to $\overline{\mathcal{W}}^\Scal$. We can then write

\begin{flalign}	
\begin{aligned}
    \textbf{given}\ \      & \{p_i,C_i\} \\
    \textbf{min}\ \   & \lambda \\ 
    \textbf{s.t.}\ \ 
    & p_i C_i^{\otimes k} \leq \lambda \overline{W}^\Scal \\ 
    & \overline{W}^\Scal \in \overline{\mathcal{W}}^\Scal,
\end{aligned}&&
\end{flalign}
where the dual affine space of the sets used in this work are explicitly presented in Fig.~\ref{fig::affineduals}.

Due to the product of variables $\lambda$ and  $\overline{W}^\Scal$,  the constraint $p_i C_i^{\otimes k} \leq \lambda \overline{W}^\Scal$ is not linear. This problem can be easily circumvented by noting that the elements of dual affine spaces have a fixed trace $\tr(\overline{W}^\Scal)$.  We can then ``absorb'' the variable $\lambda$ into $\overline{W}^\Scal$ by defining $\overline{W'}^\Scal:=\lambda W^\Scal$. 

For the case of separable testers, the primal problem can be formulated as

\begin{flalign} \label{SDP:SEPprimal}
\begin{aligned}
    \textbf{given}\ \      &\{p_i,C_i\} \\
    \textbf{max}\ \   &\sum_i p_i \tr\left(T_i^\text{SEP}\,C_i^{\otimes 2}\right) \\ 
    \textbf{s.t.}\ \ & T_i^\text{SEP}\geq0 \\
    &W^\text{SEP}:=\sum_i T_i^\text{SEP} = q W^{1\prec2} + (1-q)W^{2\prec1} \\
    & W^{1\prec2} \in {\mathcal{W}^{1\prec2}}, \quad W^{2\prec1}  \in {\mathcal{W}^{2\prec1}} \\
    & q\in[0,1],
\end{aligned}&&
\end{flalign}
where ${\mathcal{W}^{i\prec j}}$ is the set of sequential processes with slot $i$ coming before slot $j$. The SDP described in Eqs.~\eqref{SDP:SEPprimal} can also be written as
\begin{flalign} \label{SDP:SEPprimal2}
\begin{aligned}
    \textbf{given}\ \      &\{p_i,C_i\} \\
    \textbf{max}\ \   &\sum_i p_i \tr\left(T_i^\text{SEP}\,C_i^{\otimes 2}\right) \\ 
    \textbf{s.t.}\ \ & T_i^\text{SEP}\geq0 \\
    &\sum_i T_i^\text{SEP} =  W^{1\prec2} + W^{2\prec1} \\
    & \tr\left(W^{1\prec2}\overline{W}^{1\prec2}_a\right)=q, \quad  \forall a \\ & \tr\left(W^{2\prec1}\overline{W}^{2\prec1}_b\right)=1-q, \quad  \forall b \\
    & W^{1\prec2} \geq 0, \ \ W^{2\prec1}\geq0 
\end{aligned}&&
\end{flalign}
where the set $\left\{\overline{W}^{i\prec j}_l\right\}_l$ is an basis for the dual affine space of ordered processes. 

The Lagrangian of the SDP presented in Eqs.~\eqref{SDP:SEPprimal2} can be written as
\begin{align}
L &=
\sum_{i} p_i \tr\left(C_i^{\otimes 2} \; T_i^\text{SEP}\right) 
+ \sum_i \tr\left(T_i^\text{SEP} \Gamma_i\right) 
+\sum_i \tr\left[(T_i^\text{SEP}-W^{1\prec2}-W^{2\prec1})H\right] \\
&+ \sum_{a} \left[ q-\tr\left(W^{1\prec2} \overline{W}_a^{1\prec2}\right)\right] \lambda_a^{1\prec2}
+ \sum_b \left[ (1-q) -\tr\left(W^{2\prec1} \overline{W}_b^{2\prec1}\right)\right] \lambda_b^{2\prec1} \\
&+ \tr(W^{1\prec2} \sigma^{1\prec2}) + \tr(W^{2\prec1}\sigma^{2\prec1}).
\end{align}
By re-arranging terms we obtain
\begin{align}
L &=
\tr\left[T_i^\text{SEP}(p_iC_i^{\otimes 2} + \Gamma_i + H)\right]\\
&+ \tr\left[ W^{1\prec2} \left(\sigma^{1\prec2}- H - \sum_a \overline{W}_a^{1\prec2} \lambda^{1\prec2}_a \right)  \right]
+ \tr\left[ W^{2\prec1} \left(\sigma^{2\prec1} - H - \sum_b \overline{W}_b^{2\prec1} \lambda^{2\prec1}_b \right)  \right] \\
& +q\left(\sum_a \lambda^{1\prec2}_a-\sum_b \lambda^{2\prec1}_b\right) 
+ \sum_b \lambda^{2\prec1}_b.
\end{align}
Which leads to the  dual problem
\begin{flalign}
\begin{aligned}
    \textbf{given}\ \      &\{p_i,C_i\} \\
    \textbf{minimize}\ \   &\sum_b \lambda^{2\prec1}_b \\ 
    \textbf{s.t.}\ \ & \Gamma_i\geq0 \ \ \forall\, i \\
& \sigma^{1\prec2}\geq0, \ \ \sigma^{2\prec1}\geq0 \ \ \\
& q^0\geq0, \ \ q^{1\prec2}\geq0 \ \ \\
& \Gamma_i = -p_i C_i^{\otimes 2} - H, \ \ \forall\, i \\
& \sigma^{1\prec2}= H +\sum_a \overline{W}_a^{1\prec2} \lambda_a^{1\prec2}\\
& \sigma^{2\prec1} = H +\sum_b \overline{W}_b^{2\prec1} \lambda_b^{2\prec1}  \\
& \sum_b \lambda_b^{2\prec1}=\sum_a \lambda_a^{1\prec2}.
\end{aligned}&&
\end{flalign}
By removing the dummy variables we get
\begin{flalign}
\begin{aligned}
    \textbf{given}\ \      &\{p_i,C_i\} \\
    \textbf{minimize}\ \   &\sum_b \lambda^{2\prec1}_b  \\ 
    \textbf{s.t.}\ \ 
& p_i C_i^{\otimes 2}  \leq -H, \ \ \forall\, i \\
& -H \leq  \sum_a \overline{W}_a^{1\prec2} \lambda_a^{1\prec2}  \\
&  -H \leq  \sum_b \overline{W}_b^{2\prec1} \lambda_b^{2\prec1}\\
& \sum_b \lambda_b^{2\prec1}=\sum_a \lambda_a^{1\prec2}.
\end{aligned}&&
\end{flalign}
    As before we define 
$\lambda:=\sum_b\lambda_b^{2\prec1}=\sum_a\lambda^{1\prec2}_a$,
$\overline{W}^{1\prec 2}:=\sum_a \frac{\lambda_a^{1\prec2} \overline{W_a}^{1\prec 2}}{\lambda}$ and
$\overline{W}^{2\prec 1}:=\sum_b \frac{\lambda_b^{2\prec1}  \overline{W_a}^{2\prec 1}}{\lambda}$, and set $-H\mapsto H$
to obtain the simplified problem
\begin{flalign}
\begin{aligned}
    \textbf{given}\ \      &\{p_i,C_i\} \\
    \textbf{minimize}\ \   & \lambda \\ 
    \textbf{s.t.}\ \ 
& p_i C_i^{\otimes 2}  \leq H, \ \ \forall\, i \\
& H  \leq \lambda \overline{W}^{1\prec2} \\   
& H \leq \lambda \overline{W}^{2\prec1}    \\
& \overline{W}^{1\prec2} \in \overline{\mathcal{W}}^{1\prec2}\\
&\overline{W}^{2\prec1} \in \overline{\mathcal{W}}^{2\prec1}.
\end{aligned}&&
\end{flalign}
As previously explained, since the operators $\overline{W}^{1\prec2}$ and $\overline{W}^{2\prec1}$ have a fixed trace, by absorbing the coefficient $\lambda$ this problem can be straightforwardly phrased as an SDP.

\section{Computer-assisted proofs}\label{app::compassis}

In this section we provide a general algorithm that can be used to obtain a rigorous computer-assisted proof from numerical optimization packages which may use floating-point variables. Since floating-point variables use approximations to store real numbers, the constraints required by the optimization problem cannot be satisfied exactly. For instance, let $C_\texttt{float}\in\mathcal{L}(\H^I\otimes\H^O)$ be a matrix with floating-point variables which is certified by a computer to respect the quantum channel constraints, i.e.,
\begin{align}
     C_\texttt{float} &\geq0 \\
    _O C_\texttt{float} &= _{IO}C_\texttt{float} \label{eq:proj}\\
    \tr(C_\texttt{float})&=d_I.
\end{align}
Due to floating-point rounding errors, these constraints may be violated in a rigorous analysis, that is, they are satisfied only up to a numerical precision. For this reason, numerical solutions involving floating-point variables or rounding approximations may lead to accuracy problems \cite{floating_wikipedia,floating_guide}.
In order to circumvent the floating-point accuracy issue, we provide an algorithm that, given a floating-point variable matrix which satisfies the constraints of a desired set, up to some numerical precision, we construct another matrix which does not make use of floating-point and satisfies the constraints of the desired set exactly. Here, by desired set we refer to six main sets consider in this work: parallel processes, sequential processes, general processes, and their dual affine spaces.

Before proceeding, we present a useful characterization of the aforementioned sets in a unified manner in terms of projections. More precisely, all these sets can be written as: $C\in \mathcal{L}(\mathcal{H})$ belongs to the desired set $\mathcal{C}\subseteq \mathcal{L}(\mathcal{H}) $ if and only if\footnote{Note that when dual affine spaces are considered, the positivity constraints $C\geq0$ is not required.}
\begin{align} 
    C &\geq0 \label{eq:1}\\
    C &= \map{P}(C) \label{eq:2}\\
    \tr(C)&=\gamma, \label{eq:3}
\end{align}
for a suitable linear space $\H$, for some linear projection map $\map{P}:\mathcal{\H}\to\mathcal{\H}$, i.e., some map $\map{P}$ such that $\map{P}\circ\map{P}=\map{P}$ and $\map{P}(\id)=\id$, and for some normalization coefficient $\gamma$. Here, the set $\mathcal{C}$ is phrased in such a general way that it covers, for example, the set of quantum states, channels, combs, and processes, among others.

For instance, if the desired set $\mathcal{C}$ is the set of quantum channels, we have that  $\H = \H_{I}\otimes \H_{O}$ and $C\in\mathcal{C}$ if and only if
\begin{align} 
    C &\geq0 \\
    C &= \map{P}(C)=  C - _O C + _{IO}C \\ 
    \tr(C)&= \gamma = d_I.
\end{align}

If the desired set is  the set  of two-slot parallel processes $\mathcal{W}^\text{PAR}$, we have that $\H = \H_{I_1}\otimes \H_{O_1}\otimes \H_{I_2}\otimes \H_{O_2}$ and $W\in \mathcal{W}^\text{PAR}$ if and only if
\begin{align} 
    W &\geq0 \\
    W &= \map{P}^\text{PAR}(W)= _{O_1O_2}W \\
    \tr(W)&=\gamma^\text{PAR} = d_{O_1}d_{O_2}.
\end{align}
The projection maps $\map{P}^\Scal$ for the sets of processes $\Wcal^\Scal$ and $\map{\overline{P}}^\Scal$ for the sets of dual affine spaces $\overline{\Wcal}^\Scal$ used in this section are presented in Table.~\ref{tb::projections}.
    
\begin{table}[h!]
\begin{center}
{\renewcommand{\arraystretch}{2}
\begin{tabular}{| c | l | l |}               
	\multicolumn{1}{c}{} & \multicolumn{1}{c}{Processes} & \multicolumn{1}{c}{Dual affine space (Channels)}\\ 
    \hline                 
	PARALLEL   & $\ $ $\widetilde{P}^\text{PAR}(W) = _{O_1O_2}W$ & $\ $ $\widetilde{\overline{P}}^\text{PAR}(\overline{W}) = \overline{W} -_{O_1O_2}\overline{W} + _{I_1I_2O_1O_2}\overline{W}$ \\	
	SEQUENTIAL & $\ $ $\widetilde{P}^\text{SEQ}(W) = _{O_2}W-_{I_2O_2}W+_{O_1I_2O_2}W$ & $\ $ $\widetilde{\overline{P}}^\text{SEQ}(\overline{W}) = \overline{W}-_{O_2}\overline{W}+_{I_2O_2}\overline{W}-_{O_1I_2O_2}\overline{W} + _{I_1O_1I_2O_2}\overline{W}$ $\, $ \\
	GENERAL    & \begin{tabular}{l} $\ $ $\widetilde{P}^\text{GEN}(W) = _{I_1O_1O_2}W - _{I_1O_1}W + _{O_1I_2 O_2}W$ \\ \hspace*{1.5cm} $- _{I_2O_2}W + _{O_1}W + _{O_2}W - _{O_1O_2}W $ \end{tabular} & \begin{tabular}{l} $\ $ $\widetilde{\overline{P}}^\text{GEN}(\overline{W}) = \overline{W} - _{O_1}\overline{W} + _{I_1O_1}\overline{W} - _{O_2}\overline{W} + _{I_2O_2}\overline{W} $ \\ \hspace*{1.5cm} $-  _{O_1 I_2 O_2}\overline{W} -  _{I_1 O_1 O_2}\overline{W} + _{O_1O_2}\overline{W} + _{I_1O_1I_2 O_2}\overline{W}$\end{tabular} \\
    \hline 
\end{tabular}
}
\end{center}
\caption{Projectors onto the linear space spanned by  parallel, sequential, and general processes and their respective dual affine spaces. In all these cases, $\H = \H_{I_1}\otimes \H_{O_1}\otimes \H_{I_2}\otimes \H_{O_2}$. In addition, we remark that the trace constraint of Eq.~\eqref{eq:3} for processes is $\tr(W)=d_{O_1}d_{O_2}$ and for their dual affine spaces, we have $\tr(\overline{W})=d_{I_1}d_{I_2}$.}
\label{tb::projections}
\end{table}

We now present Algorithm 1, which takes a linear operator $C_\texttt{float}$ respecting the conditions of a set $\mathcal{C}$ described by Eqs.~\eqref{eq:1}-\eqref{eq:3} up to numerical precision and provide an operator $C_\texttt{OK}$ which respects the conditions of $\mathcal{C}$ exactly. Also, all the steps of our algorithm can be done without approximations or the use of numerical floating-point variables.
\\

\paragraph*{\textbf{\emph{Algorithm 1:}}}
\begin{enumerate}
    \item 
    \texttt{Construct the non-floating-point matrix 
    $C_{\text{frac}}$ 
    by truncating the matrix
    $C_\text{float}$  } \\ 
    This allows us to work with fractions and to avoid numerical imprecision.
    \item
    \texttt{Define the matrix} $\displaystyle{C:=\frac{C_\texttt{frac}+\left(C_\texttt{frac}\right)^\dagger}{2}}$
    \texttt{to obtain a self-adjoint matrix $C$} \\ 
    Ensures that we are dealing with self-adjoint matrices
    \item \texttt{Project $C$ into a valid subspace and obtain $\map{P}(C)$} \\ Ensures that the operator is in the valid linear subspace.
    \item  \texttt{Find a coefficient $\eta$ such that $\map{D}_\eta\left(\map{P}(C)\right):=\eta \map{P}(C) + (1-\eta)\id$ is positive semidefinite} \\
    Ensures positivity without leaving the valid subspace.
    \item  \texttt{Output the operator $\displaystyle{C_\texttt{OK}= \gamma \frac{\map{D}_\eta(\map{P}(C))}{ \tr[\map{D}_\eta(\map{P}(C))] }}$
    which lies in $\mathcal{C}$}\\
    Ensures the trace condition, preserving positivity and without leaving the valid subspace.
\end{enumerate}

One way to complete step 4 is to start with $\eta=1$ and check if the operator $C$ is already positive semidefinite. If $C$ is not positive semidefinite, we can slowly decrease the value of $\eta$ and check if $\map{D}_\eta(C)$ is positive definite. Checking if a matrix is positive semidefinite can be done efficiently by implementing the Cholesky decomposition algorithm and checking whether the algorithm leads to a valid Cholesky decomposition. 
  
One can verify that the operator $C_\texttt{OK}$ provided by the algorithm described above necessarily belongs to the desired valid set $\mathcal{S}$ with the aid of the following theorem.

\begin{theorem}
    Let $\map{P}:\mathcal{L}(\H)\to \mathcal{L}(\H)$ be a linear projector i.e., $\map{P}\circ\map{P}=\map{P}$, which respects $\map{P}(\id)=\id$. Let $\map{D}_\eta:\mathcal{L}(\H)\to \mathcal{L}(\H)$ be an affine map defined by $\map{D}_\eta(C):=\eta C + (1-\eta) \id$. It holds that
    
    \begin{equation}
        \map{D}_\eta\left(\map{P}(C)\right) =
        \map{P}\left(\map{D}_\eta\left(\map{P}(C)\right) \right)
    \end{equation}
\end{theorem}

\begin{proof}
    \begin{align}
        \map{P}\left(\map{D}_\eta\left(\map{P}(C)\right) \right) 
        &=\map{P}\left(\eta \map{P}(C) + (1-\eta) \map{P}(\id) \right) \\
        &=\eta \map{P}\left(\map{P}(C)\right) + (1-\eta) \map{P}\left(\map{P}(\id) \right) \\
        &=\eta \map{P}(C) + (1-\eta) \id \\
        &= \map{D}_\eta\left( \map{P}(C) \right).
    \end{align}
\end{proof}

Algorithm 1 allows us to obtain upper bounds for the maximal probability of discriminating an ensemble of quantum channels. For the case in which the desired set $\mathcal{C}$ is the set of dual affine spaces of processes $\overline{\Wcal}^\Scal$ for some strategy $\Scal$, $C_\texttt{float}$ is the floating-point matrix of a dual affine $\overline{W}^\Scal_\texttt{float}$, that can be obtained using numerical convex optimization packages to solve the dual problem SDP, and map $\map{P}$ is one of the projection maps $\map{\overline{P}}^\Scal$, then Algorithm 1 will return a matrix $\overline{W}^\Scal_\texttt{OK}$ that satisfies the constraints of the set $\overline{\Wcal}^\Scal$ exactly. A rigorous upper bound on the maximal probability for discriminating the ensemble $\{p_i,C_i\}_i$ is then given by the value $p_\texttt{upper}$ such that $p_iC_i^{\otimes k}\leq p_\texttt{upper} \overline{W_\texttt{OK}}$ for all $i$. Note that if the channels $C_i$ are also represented with floating-point variables, one can also use Algorithm 1 to obtain exact channels $C_{i,\texttt{OK}}$.
    
In order to calculate lower bounds, we can use the primal SDP to obtain a set of $\{T_{i,\texttt{float}}\}_{i=1}^N$ which satisfies the conditions of some desired class of tester up to some numerical precision. To tackle this situation, we present an algorithm to obtain a set of operators $\{T_{i,\texttt{OK}}\}_{i=1}^N$ which satisfies the tester constraints exactly. Note that this algorithm also works for positive-operator valued measures (POVMs), instruments, and super-instruments, among others.
\\

\paragraph*{\textbf{\emph{Algorithm 2:}}}
\begin{enumerate}
    \item 
    \texttt{Construct the non-floating-point matrix 
    $T_{i, \text{frac}}$ 
    by truncating the matrix
    $T_{i, \text{float}}$  } \\ 
    This allows us to work with fractions and to avoid numerical imprecision.
    \item 
    \texttt{Define the matrices} $\displaystyle{T_i:=\frac{T_{i, \texttt{frac}}+\left(T_{i, \texttt{frac}}\right)^\dagger}{2}}$
    \texttt{to obtain self-adjoint matrices $T_i$} \\ 
    Ensures that we are dealing with self-adjoint matrices.
    \item \texttt{Project $W:=\sum_{i=1}^N T_i$ into a valid subspace and obtain $\map{P}(W)$ } \\ Ensures the operator $W$ is in the valid linear subspace.
    \item \texttt{Define the extra-outcome tester element $T_\varnothing:=\map{P}(W)-W$ } \\ Useful step to later ensure the normalization constraints.
    \item  \texttt{Find a coefficient $\eta$ such  $\map{D}_\eta(T_\varnothing)\geq0$
    and 
    $\map{D}_\eta(T_i)
    \geq 0$ holds for every $i$} \\
    Ensures positivity of all tester elements.
    \item  \texttt{Define $W_\eta:=\left(\sum_{i=1}^N \map{D}_\eta(T_i)\right) + \map{D}_\eta(T_\varnothing)$} \\
    Defines a positive semidefinite operator such that $
    W_\eta = \map{P}(W_\eta)$.
    \item  
    \texttt{Output the set 
    $\displaystyle{T_\text{OK}:=
    \left\{
    \gamma \frac{  \map{D}_\eta(T_i)  +  \frac{\map{D}_\eta(T_\varnothing)}{N} }  { \tr(W_\eta)}
    \right\}_i}$
    which is a valid tester} \\
    Equally distributes the tester element $\map{D}_\eta(T_\varnothing)$ between elements indexed by $i$.
\end{enumerate}

Similarly to algorithm 1, one can verify that the set $T_\texttt{OK}$ is a valid tester.
\begin{theorem}
    The operator $T_\texttt{OK}$ defined in step 6 of Algorithm 2 is a valid tester.
\end{theorem}

\begin{proof}
By construction all tester elements 
\begin{equation}
    T_{i\texttt{OK}}:=\gamma \frac{  \map{D}_\eta(T_i)  +  \frac{\map{D}_\eta(T_\varnothing)}{N} }  { \tr(W_\eta)}
\end{equation} 
are positive semidefinite, we then need to show that $W_\texttt{OK}:=\sum_i T_{i, \texttt{OK}}$ respects $\map{P}(W_\texttt{OK})= W_\texttt{OK}$ and $\tr(W_\texttt{OK})=\gamma$. For that, note that
\begin{align}
    W_\texttt{OK} =& \gamma 
\sum_{i=1}^N \frac{  \map{D}_\eta(T_i)  +  \frac{\map{D}_\eta(T_\varnothing)}{N} }  { \tr(W_\eta)}
 \\
=& \gamma \frac{W_\eta}{\tr(W_\eta)}.
\end{align}
We can then guarantee that $\tr(W_\texttt{OK})=\gamma$ and
\begin{align}
    W_\eta &=\left[\sum_{i=1}^N\eta T_i + (1-\eta)\id\right]
    + \eta T_\varnothing + (1-\eta) \id \\
    &= \eta W + (1+\eta)N \id + \eta \map{P}(W)-\eta W + (1-\eta) \id \\
    &= \eta\map{P}(W) + (1-\eta)(N+1)\id \\
    &= \eta\map{P}(W) + (1-\eta)(N+1)\map{P}(\id) \\
    &= \map{P}(W_\eta) .
\end{align}
\end{proof}

Algorithm 2 allows us to obtain lower bounds for the maximal probability of discriminating an ensemble of quantum channels. A floating-point set of matrices $\{T^\Scal_{i, \texttt{float}}\}_i$, can be obtained via numerical convex optimization packages to solve the primal problem SDP. Then Algorithm 2 will return a set of matrices $T_\texttt{OK}^\Scal$  that satisfies the constraints of the set $\mathcal{T}^\Scal$ exactly. A rigorous lower bound on the maximal probability for discriminating ensemble $\{p_i,C_i\}_i$ is then given by the value $p_\texttt{lower}=\sum_{i=1}^N p_i \tr\left(C_i^{\otimes k} T_{i,\texttt{OK}}^\Scal\right)$.

We have implemented the algorithms presented in this section and the remaining code necessary for the calculation of the upper- and lower bounds presented in this paper. All code has been uploaded to an online repository~\cite{githubMTQ}. The SDP optimization was implemented in MATLAB™ using the package cvx~\cite{cvx} and tested independently with the solvers MOSEK~\cite{mosek}, SeDuMi~\cite{sedumi}, and SDPT3~\cite{sdpt3}. The computer-assisted proof step used to obtain the exact upper and lower bounds was implemented in Mathematica™. All our code can be freely used, edited, and distributed under the MIT license \cite{MIT_license}.

\section{Sampling general channels and the typicality of the hierarchy between discrimination strategies}\label{app::randchannels}

Our method for generating a general channel goes as follows:
\begin{enumerate}
    \item Fix input dimension $d_I$ and output dimension $d_O$.
    \item Uniformly sample a positive semidefinite matrix $A$ of size ($d_Id_O$)-by-($d_Id_O$), according to the Hilbert-Schmidt measure. This can be done, for example, using the function \texttt{RandomDensityMatrix} of the freely distributed MATLAB toolbox QETLAB~\cite{qetlab}.
    \item Define $C$ to be the projection of $A$ on the subspace of valid quantum channels, according to
    \begin{equation}
        C = A - {_O} A + \frac{\id}{d_O}.
    \end{equation}
    \item Check whether $C$ is a positive semidefinite matrix. If not, discard $C$ and repeat the process. If yes, than $C$ represents the Choi operator of a valid quantum channel $\widetilde{C}:\Lcal(\Hcal^I)\to\Lcal(\Hcal^O)$.
\end{enumerate}

We have sampled $100,000$ pairs of general qubit-qubit channels using this method and computed, using our SDP methods, the maximal probability of discriminating these channels in an ensemble where both channels are equally probable, using parallel, sequential, separable, and general strategies. Our results are summarized in Table~\ref{tb::randchannels}. The first column denotes between which strategies a gap was found and the second column denotes how many of the $100,000$ pairs of channels had such gap.

\begin{table}[h!]
\begin{center}
\begin{tabular}{| c | c |}
		\hline                             
		Strategy gap & $\quad$ Number of pairs of channels $\quad$ \\
		& (out of 100\,000) \\
		\hline 
		$P^\text{PAR}<P^\text{SEQ}$ & $99\,955$ \\
		$P^\text{SEQ}<P^\text{SEP}$ & $99\,781$  \\
		$P^\text{SEP}<P^\text{GEN}$ & $94\,026$ \\
		$\quad$ $P^\text{PAR}<P^\text{SEQ}<P^\text{SEP}<P^\text{GEN}$ $\quad$ & $94\,015$ \\
        \hline
\end{tabular}
\end{center}
\caption{The first column denotes between which strategies of channel discrimination a gap in performance was found and the second column denotes how many of the $100,000$ pairs of channels that were sampled demonstrated such a gap.}
\label{tb::randchannels}
\end{table}

In particular, the last line of Table~\ref{tb::randchannels}, which shows that a strict hierarchy  $P^\text{PAR}<P^\text{SEQ}<P^\text{SEP}<P^\text{GEN}$ between all four strategies was found by $94,015$ pairs of channels, implies that our method has around $94\%$ probability of generating a pair of qubit-qubit channels that showcases this phenomenon.

\begin{figure}
\begin{center}
	\includegraphics[width=\columnwidth]{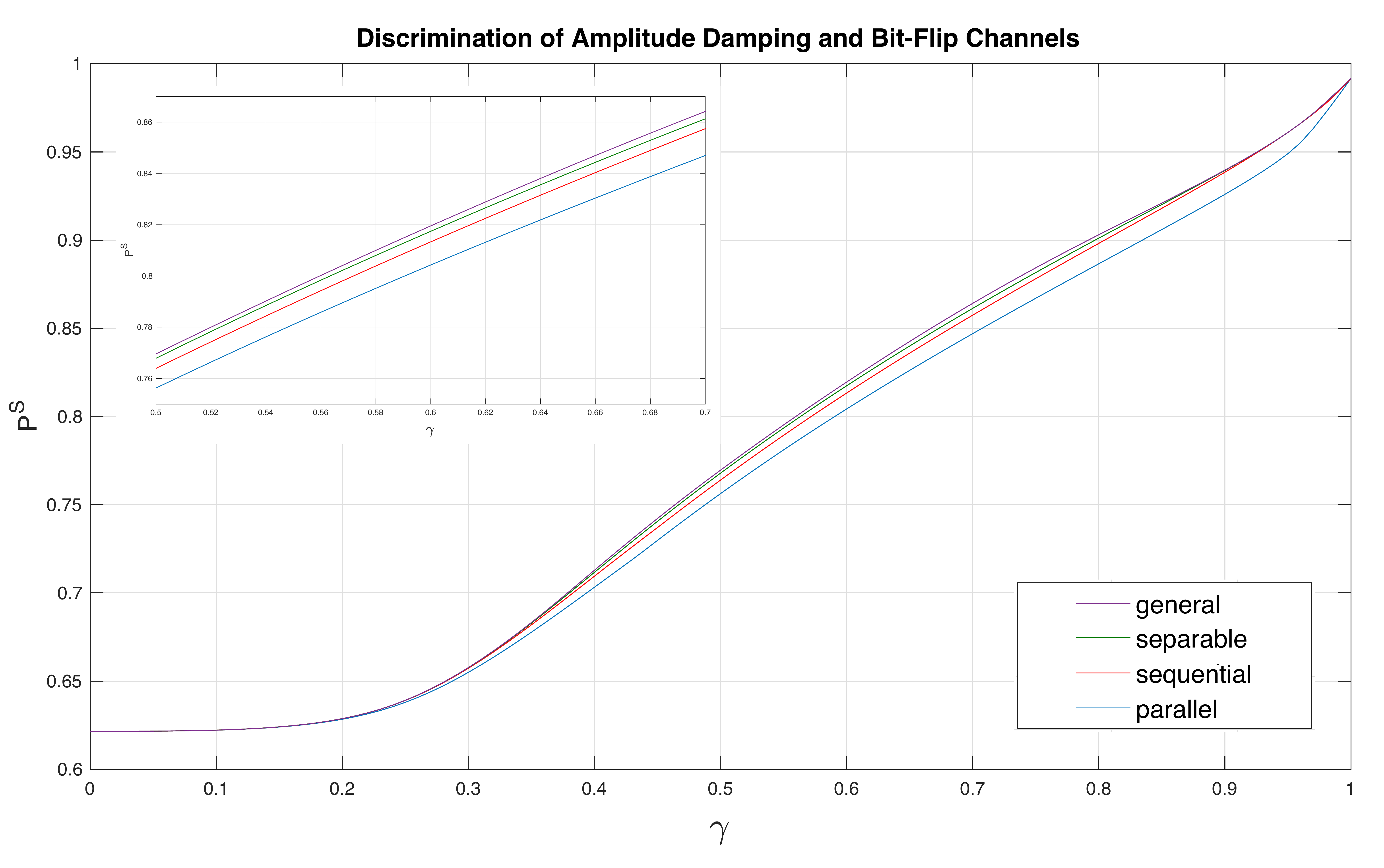}
	\caption{Probability of successfully discriminating an amplitude damping channel and a bit-flip channel, in an equiprobable ensemble, using $k=2$ copies. The value of the decay parameter of the amplitude damping channel varies with the interval $\gamma\in[0,1]$, while the flipping parameter of the bit-flip channel is fixed at $\eta=0.87$. The four curves represent parallel, sequential, separable, and general strategies of channel discrimination. A clear gap between all four strategies is clearly visible in the picture-in-picture plot, with $\gamma\in[0.5,0.7]$.}
\label{fig::ampdamp_bitflip}
\end{center}
\end{figure}

For the case of discriminating between amplitude damping channels and bit-flip channels, in order to show that the phenomenon of the advantage between different strategies is not unique to a specific choice of parameters, we plot on Fig.~\ref{fig::ampdamp_bitflip} the probability of successful discrimination between an amplitude damping channel with decay parameter $\gamma\in[0,1]$ and a bit-flip channel with fixed flipping parameter $\eta=0.87$. A clear gap between all four strategies can be clearly seen on the zoomed picture-in-picture, which plots only $\gamma\in[0.5,0.7]$. Similar plots can be obtained for different values of $\eta$.

It is also true that a strict hierarchy between strategies of channel discrimination can be found when discriminating among two amplitude damping channels, in an equiprobable ensemble, with different decay parameters. Using our methods, we have calculated the probability of success for all four strategies, and would like to point out one interesting case of discrimination between one amplitude damping channel with $\gamma_1=0.37$ and another with $\gamma_2=0.87$, which gives 

\begin{equation}
\begin{alignedat}{3}
  \frac{8101}{10000} < P^\text{PAR} < \frac{8102}{10000} & && && \\
  & < \frac{8161}{10000} < P^\text{SEQ} < \frac{8162}{10000} && && \\
  & && < \frac{8166}{10000} < P^\text{SEP} < \frac{81665}{100000} && \\
  & && && < \frac{8167}{10000} < P^\text{GEN} < \frac{8168}{10000}.
\end{alignedat}
\end{equation}


Here, we confirm that there exists advantage in the discrimination of amplitude damping channels using sequential strategies over parallel strategies.  Furthermore, we show that the case of discrimination among two amplitude damping channels is also an example of a complete hierarchy among all four strategies.


\begin{thebibliography}{66}
\bibitem [{fn:({\natexlab{a}})}]{fn::cj}%
  \BibitemOpen
  \href@noop {} {\bibinfo {title} {{The Choi-Jamio\l{}kowski isomorphism is a
  one-to-one correspondence between linear maps
  $\widetilde{\Lambda}:\mathcal{L}(\mathcal{H}^I)\to\mathcal{L}(\mathcal{H}^O)$
  and linear operators $L\in\mathcal{L}(\mathcal{H}^I\otimes\mathcal{H}^O)$
  defined by $L \coloneqq (\widetilde{\id}\otimes\widetilde{\Lambda})(\Phi^+)$,
  where $\Phi^+=\sum_{ij}\ket{ii}\bra{jj} \in
  \mathcal{L}(\mathcal{H}^I\otimes\mathcal{H}^I)$, with $\{\ket{i}\}$ being an
  orthonormal basis}}} \BibitemShut {NoStop}%
\bibitem [{fn:({\natexlab{b}})}]{fn::linkprod}%
  \BibitemOpen
  \href@noop {} {\bibinfo {title} {{Let $F\in \mathcal{L}(\mathcal{H}^A \otimes
  \mathcal{H}^B)$ and $G\in \mathcal{L}(\mathcal{H}^B\otimes \mathcal{H}^C)$ be
  two linear operators. We define the link product between them as
  $F*G\coloneqq \tr_B[(F^{AB}\otimes \id^C)(\id^A \otimes {G^{BC}}^{T_B})]$,
  where $\cdot^{T_B}$ is the transposition in the computational basis of
  $\mathcal{H}^B$. If $C_i\in \mathcal{L}(\mathcal{H}^I\otimes\mathcal{H}^O)$
  is the Choi operator of a quantum channel
  $\widetilde{C}_i:\mathcal{L}(\mathcal{H}^I)\to \mathcal{L}(\mathcal{H}^O)$,
  then for every state $\rho\in\mathcal{L}(\mathcal{H}^I)$ we have that
  $\widetilde{C}_i(\rho) = C_i*\rho = \tr_I
  \left[C_i\,(\rho^T\otimes\id^O)\right]\in\mathcal{L}(\mathcal{H}^O)$}}}
  \BibitemShut {NoStop}%
\bibitem [{fn:({\natexlab{c}})}]{fn::transpose}%
  \BibitemOpen
  \href@noop {} {\bibinfo {title} {{Notice that we can drop the transpose over
  $T_i$ in the expression of the maximal probability of success since the set
  of testers is equal to the set of testers with a transposition, not affecting
  the optimization}}} \BibitemShut {NoStop}%
\bibitem [{fn:({\natexlab{d}})}]{fn::affinesep}%
  \BibitemOpen
  \href@noop {} {\bibinfo {title} {{For the case of separable testers, instead
  of imposing that $\overline{W}^\Scal$ lies in the dual affine set of
  separable processes, we should impose that $p_iC_i^{\otimes 2}\leq H$ with
  $H\leq \overline{W}^{1\prec 2}$ and $H\leq \overline{W}^{2\prec 1}$ where
  $\overline{W}^{i\prec j}$ lies in the dual affine set of sequential processes
  $\mathcal{W}^{i\prec j}$. See the Suppl. Material for details. Notice that
  our formulation of the dual problem is analogous to the one presented in
  Ref.~\cite{chiribella12b,chiribella16} and may be seen as a generalization of
  the Yuen-Kennedy-Lax bound for state discrimination~\cite{yuen75}}}}
  \BibitemShut {NoStop}%
\bibitem [{fn:({\natexlab{e}})}]{fn::ampdamp}%
  \BibitemOpen
  \href@noop {} {\bibinfo {title} {{The action of an amplitude damping channel
  on a qubit state is given by
  $\widetilde{C}_\text{AD}(\rho)=K_0\rho\,K^\dagger_0+K_1\rho\,K^\dagger_1$,
  where $K_0= \ketbra{0}{0} + \sqrt{1-\gamma}\ketbra{1}{1}$ and $K_1=
  \sqrt{\gamma}\ketbra{0}{1}$}}} \BibitemShut {NoStop}%
\bibitem [{fn:({\natexlab{f}})}]{fn::bitflip}%
  \BibitemOpen
  \href@noop {} {\bibinfo {title} {{The action of a bit-flip channel on a qubit
  state is given by $\widetilde{C}_\text{BF}(\rho)=\eta\,\rho+(1-\eta)X\rho X$,
  where $X$ is the Pauli $X$ operator}}} \BibitemShut
  {NoStop}%
\end{thebibliography}
\end{document}